\documentclass{llncs}

\usepackage{amsmath}
\usepackage{amssymb}
\usepackage {graphicx}
\begin{document}

\title{On the Minimum Size of a Contraction-Universal Tree}
\author{Olivier Bodini}
\institute{\obeylines LIP, \'Ecole Normale Sup\'erieure de Lyon,
46 All\'ee d'Italie, 69364 Lyon Cedex 05, France.} \maketitle
\pagestyle{myheadings} \markboth{\hfill{\sc WG (2002)}} {{\sc WG
(2002)}\hfill}

\begin{abstract}

A tree $T_{uni} $ is $m$-\textit{universal} for the class of trees
if for every tree $T$ of size $m$, $T$ can be obtained from
$T_{uni} $ by successive contractions of edges. We prove that a
$m$-universal tree for the class of trees has at least $m\ln (m) +
(\gamma - 1)m + O(1)$ edges where $\gamma $ is the Euler's
constant and we build such a tree with less than $m^c$ edges for a
fixed constant $c = 1.984...$
\end{abstract}

\section{Introduction}

What is the minimum size of an object in which every object of
size $m$ embeds? Issued from the category theory, questions of
this kind appeared in graph theory. For instance, R. Rado~
\cite{Ra} proved the existence of an "initial countable graph".
Recently, Z. F\"{u}redi and P. Komj\`{a}th \cite{FK} studied a
connected question.

We use here the following definition~: given a sub-class $C$ of
graphs (trees, planar graphs, etc.), a graph $G_{uni} $ is
$m$-\textit{universal} for $C$ if for every graph $G$ of size $m$
in $C, G$ is a minor of $G_{uni} ,$ i.e. it can be obtained from
$G_{uni}$ by successive contractions or deletions of edges.

Inspired by the Robertson and Seymour work \cite{RS} on graph
minors, P. Duchet asked whether a polynomial bound in $m$ could be
found for the size of a $m$-universal tree for the class of trees.
We give here a positive sub-quadratic answer.

From an applied point of view, such an object would possibly
allows us to define a tree from the representation of its
contraction.

The main results of this paper are the following theorems which
give bounds for the minimum size of a $m$-universal tree for the
class of trees~:

\begin{theorem}\label{th1} A $m$-universal tree for the class of trees has
at least $m\ln (m) + (\gamma - 1)m + O(1)$ edges where $\gamma$ is
the Euler's constant.
\end{theorem}

\begin{theorem} \label{th2} There exists a $m$-universal tree $T_{uni} $ for the
class of trees with less than $m^c$ edges for a fixed constant $c
= 1.984...$
\end{theorem}
Our proof follows a recursive
construction where large trees are obtained by some amalgamation
process involving simpler trees. With this method, the constant
$c$ could be reduced to 1.88... but it seems difficult to improve
this value.

We conclude the paper with related open questions.

\section{Terminology}

Our graphs are undirected and simple (with neither loops nor
multiple edges). We denote by $G(V,E)$ a graph (its vertex set is
$V(G)$ and its edge set is $E(G)$ (a subset of the family of all
the $V(G)$-subsets of cardinality 2)). Referring to C. Thomassen
\cite{Th}, we recall some basic definitions that are useful for
our purpose:

We denote by $P_n $ the path of size $n.$

If $x$ is a vertex then $d(x),$ the \textit{degree} of $x,$ is the number of edges incident to
$x.$

Let $e$ be an edge of $E(G)$, the graph denoted by $G - e$ is the
graph on the vertex set of $G$, whose edge set is the edge set of
$G$ without $e$. We call classically this operation
\textit{deletion}.

Let $e = \left\{ {a,b} \right\}$ be an edge of $G(V,E)$, we name
\textit{contraction of }$G$\textit{ along }$e,$ the graph denoted by $G / e = H(V',E')$, with ${V}' = \left( {V /
\left\{ {a,b} \right\}} \right) \cup \left\{ c \right\}$ where $c$ is a new
vertex and $E'$ the edge set which contains all the edges of the sub-graph
$G_1 $ on $V / e$ and all the edges of the form $\left\{ {c,x} \right\}$ for
$\left\{ {a,x} \right\}$ or $\left\{ {b,x} \right\}$ belonging to $E$.

We say that $H$ is a \textit{minor} of $G$ if and only if we can
obtain it from $G$ by successively deleting and~/or contracting
edges, in an other way, we can define the set $M(G)$ of minors of
$G$ by the recursive formula~:

\[
M\left( G \right) = G \cup \left( {\bigcup\limits_{e \in E\left( G \right)}
{M\left( {G / e} \right)} } \right) \cup \left( {\bigcup\limits_{e \in
E\left( G \right)} {M\left( {G - e} \right)} } \right)
\]

The notion of minor induces a partial order on graphs. We write
$A\preceq B$ to mean "$A$ is a minor of $B$".

For technical reasons, we prefer to use the size of a tree (edge number)
rather than its order (vertex number).

Finally, let us recall that, a graph $G_{uni} $ is $m$-\textit{universal} \textit{for a sub-class }$C$ of graphs if for
every element $G$ of $C$ with $m$ edges$, G$ is a minor of $G_{uni} $.

\section{A Lower Bound}

In this section, we prove that a $m$-universal tree $T_{uni} $ for
the trees has asymptotically at least $m\ln (m)$ edges. We use the
fact that $T_{uni} $ has to contain all spiders of size $m$ as
minors. A \textit{spider }$S$\textit{ on a vertex }$w$ is a tree
such that $\forall v \in V\left( S \right)\backslash \left\{ w
\right\}, d(v) \le 2$. We denote the spider constituted by paths
of lengths $1 \le m_1 \le ... \le m_k $ by $Sp(m_1 ,...,m_k )$
(Fig.1).

\begin{figure}[htbp]
\centerline{\includegraphics[width=1.75in,height=0.79in]{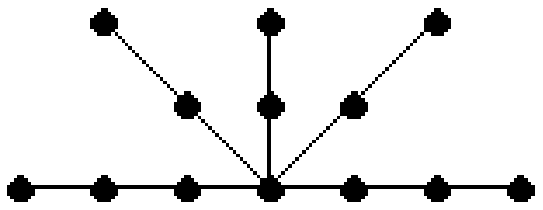}}
\label{fig1}
\begin{center}
Fig.1. $Sp(2,2,2,3,3)$
\end{center}
\end{figure}

\begin{definition} Let $T$ be a tree, we denote by $\partial
T$ the subtree of $T$ with $V(\partial T) = V(T)\backslash A$,
where $A$ is the set of the leaves of $T$. Also, we denote by
$\partial ^k$ the $k$-th iteration of $\partial $.
\end{definition}

\begin{lemma}\label{lm1} $Sp(m_1 ,...,m_k )\preceq T$ involves that $\partial Sp(m_1 ,...,m_k )\preceq \partial T$.
 Moreover, if for all $i$,
$m_i = 1$ then $\partial Sp(m_1 ,...,m_k )$ is a vertex.
Otherwise, put $a$ the first value such that $m_a > 1$, we have
$\partial Sp(m_1 ,...,m_k )= Sp(m_a - 1,...,m_k - 1)$ excepted for
$k = 1$, in this last case we have $\partial Sp(m_1 ) = Sp(m_1 -
2)$.
\end{lemma}

\begin{proof} This just follows from an observation.\qed
\end{proof}

\begin{lemma}\label{lm2} For every tree $T$, $Sp(m_1 ,...,m_k )\preceq T \Rightarrow T$ has at least $k$ leaves.
\end{lemma}

\begin{proof} Trivial.\qed
\end{proof}

\begin{theorem}\label{th3} A $m$-universal tree $T_{uni} $ for the class of trees
has at least $\sum\limits_{i = 1, i \ne 2}^m {\left\lfloor
{\frac{m}{i}} \right\rfloor } $ edges.
\end{theorem}

\begin{proof}A $m$-universal tree $T_{uni} $ for the class of trees has to contain as
minors all spiders of size $m$. So, for all $p$ it contains as
minors the spiders $Sp(p,...,p)$ where we have $\left\lfloor
{\frac{m}{p}} \right\rfloor $ times the letter $p$. By the lemma
\ref{lm1}, for all $p \leq \frac{m}{2}$, $Sp(1,...,1)\preceq
\partial ^{p - 1}T_{uni} $ and if $m$ is odd, $Sp(1)\preceq
\partial ^{\left\lfloor {\frac{m}{2}} \right\rfloor-1}T_{uni}$.
Moreover, it is clear that the terminal edges of the $\partial
^pT_{uni} $ constitute a partition of $T_{uni} $. By the lemma
\ref{lm2}, this involves that $T_{uni} $ has at least
$\sum\limits_{
 i = 1}^{\left\lfloor {\frac{m}{2}} \right\rfloor } {\left\lfloor
{\frac{m}{i}} \right\rfloor }$
 edges if $m$ is even and $1+\sum\limits_{
 i = 1}^{\left\lfloor {\frac{m}{2}} \right\rfloor} {\left\lfloor
{\frac{m}{i}} \right\rfloor}$ edges if $m$ is odd. An easy
 calculation proves that these values are always equal to $\sum\limits_{i = 1, i \ne 2}^m {\left\lfloor
{\frac{m}{i}} \right\rfloor}$.\qed
\end{proof}

\begin{proof} (of the theorem \ref{th1}) it follows from the usual estimate
$\sum\limits_{i = 1}^n {\frac{1}{i}} \sim \ln \left( n \right) +
\gamma + O\left( {\frac{1}{n}} \right)$ and the inequality
$\sum\limits_{i = 1, i \ne 2}^m {\left\lfloor {\frac{m}{i}}
\right\rfloor} \ge 1 + \sum\limits_{
 i = 1,
 i \ne 2}^{m - 1} {\left( {\frac{m}{i} - 1} \right)} $.\qed
\end{proof}

What the above proof shows, in fact, is the following :

\begin{corollary}A minimum $m$-universal spider for the class of
spiders has $\sum\limits_{
 i = 1,
 i \ne 2}^m {\left\lfloor {\frac{m}{i}} \right\rfloor } $ edges.
\end{corollary}

\begin{proof}The spider $Sp\left( {\left\lfloor {\frac{m}{m}} \right\rfloor
,\left\lfloor {\frac{m}{m-1}} \right\rfloor ,...,\left\lfloor
{\frac{m}{2}} \right\rfloor,\left\lceil {\frac{m}{2}}\right\rceil
} \right)$ is clearly a $m$-universal spider of size
$\sum\limits_{
 i = 1,
 i \ne 2}^m {\left\lfloor {\frac{m}{i}} \right\rfloor } $ for the class of
spiders, and by theorem \ref{th3} it is a minimum value.\qed
\end{proof}

\section{The Main Stem}

In the sequel, we deal with\textit{ rooted graph}, i.e. graph $G$
where we can distinguish a special vertex denoted by $r(G)$,
called the \textit{root}. Conventionally, any contracted graph
${G}'$ of same rooted graph $G$ will be rooted at the unique
vertex which is the image of the root under the contraction
mapping, we say in this case that the rooted graph ${G}'$ is a
\textit{rooted contraction} of $G$. Note that, the contraction
operator suffices to obtain all minor trees of a tree. So, we can
now define the following new notion for sub-classes of rooted
trees~: a rooted tree $T_{uni} $ is\textit{ strongly $m$-universal
for a sub-classes }$C$\textit{ of rooted trees} if for every
rooted tree $T$\textit{ in }$C$ of size $m, T$ is a rooted
contraction of $T_{uni} $. The concept of root is introduced to
avoid problems with graph isomorphisms that, otherwise would
greatly impede our inductive proof.

For every edge $e$ of a tree $T$, the forest $T\backslash e$ has two connected
components. We call $e$\textit{-branch}, denoted by $B_e $, the connected component of
${T}'$ which does not contain $r\left( T \right)$, we define the root of
$B_e $ as $e \cap V\left( {B_e } \right).$

A\textit{ main stem }of a rooted tree of size $m$ is defined as a
path $P$ which is issued from the root and such that for all
$e$-branches $B_e $ with $e \notin E\left( C \right)$, we have
$\left| {E\left( {B_e } \right)} \right| < \left\lfloor
{\frac{m}{2}} \right\rfloor $ (Fig.2).

\begin{figure}[htbp]
\centerline{\includegraphics[width=1.64in,height=1.06in]{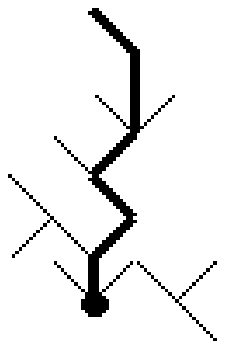}}
\label{fig2}
\begin{center}
Fig.2. A main stem in bold
\end{center}
\end{figure}

The following lemma suggests the procedure which will be used to find a
sub-quadratic upper bound for universal trees. Roughly speaking, it endows
every tree with some recursive structure constructed with the help of main
stems.

\begin{lemma} Every rooted tree has a main stem.
\end{lemma}

\begin{proof}By induction on the size of the rooted tree. Let $T$ be a rooted tree, if $T$ has
one or two edges, it is trivial. Otherwise let us consider the
sub-graph $T\backslash r\left( T \right)$, which is a forest. We
choose a connected component $T_1 $ with maximum size and we
denote by $b_1 $ the unique vertex of $T_1$ which is adjacent to
$r(T)$. Tree $T_1 $, rooted in $b_1 $, has, by the induction
hypothesis, a main stem $B.$ Then the path $\left( {V\left( B
\right) \cup \left\{ {r\left( T \right)} \right\},E\left( B
\right) \cup \left\{ {\left\{ {r\left( T \right),b_1 } \right\}}
\right\}} \right)$ is a main stem of $T$.\qed
\end{proof}

\begin{remark} A tree may possess in general several main stems. Let
us notice also that a main stem is not necessarily one of the
longest paths which contain the root.
\end{remark}

\section{The Upper Bound}

We need some new definitions. A \textit{rooted brush} (Fig.3) is a
rooted tree such that the vertices of degree greater than 2 are on
a same path $P$ issued from the root.

\begin{figure}[htbp]
\centerline{\includegraphics[width=1.35in,height=1.22in]{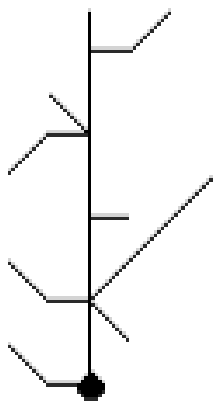}}
\label{fig3}
\begin{center}
Fig.3. A rooted brush
\end{center}
\end{figure}

A \textit{rooted comb} $X$ (Fig.4) is a rooted brush with $d\left(
{r\left( X \right)} \right) \le 2$ and $\forall v \in V\left( X
\right)$, $d\left( v \right) \le 3$.

\begin{figure}[htbp]
\centerline{\includegraphics[width=1.35in,height=1.13in]{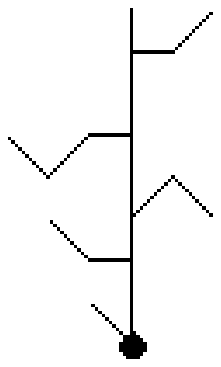}}
\label{fig4}
\begin{center}
Fig.4. A rooted comb
\end{center}
\end{figure}

The \textit{length of a rooted comb} corresponds to the length of the longest path $P$ issued from the root
which contains all vertices of degree greater than 2.

To obtain an upper bound, we consider two building processes : the
first one, a brushing $M_B $, maps rooted trees with a main stem
into rooted brushes, the second one, a ramifying $M_T $, consists
in obtaining a sequence of rooted trees, assuming that we have an
increasing sequence of rooted combs. We note $M_T^k $ the $k$-th
element of the sequence. These building processes will possess the
following fundamental property:

\begin{property} \label{pr1} Let $\left( {T,\sigma } \right)$ a rooted tree with a
main stem $\sigma $ and $\left( {X_n } \right)_{n \in \mathbb{N}}
$ a sequence of rooted combs :

\[
\left( {\forall {T}'\preceq T,M_B \left( {{T}',\sigma }
\right)\preceq X_{\left| {E\left( {T}' \right)} \right|} } \right)
\Rightarrow T\preceq M_T^{\left| {E\left( T \right)} \right|}
\left( {\left( {X_n } \right)_{n \in \mathbb{N}} } \right).
\]
\end{property}

\begin{lemma} If building processes verify the property \ref{pr1} and
if for all $i$, the rooted comb $X_i $ is strongly $i$-universal
for the class of rooted brushes then the rooted tree $M_T^m \left(
{\left( {X_n } \right)_{n \in \mathbb{N}} } \right)$ is strongly
$m$-universal for the class of rooted trees.
\end{lemma}

\begin{proof}It is just an interpretation of the property.\qed
\end{proof}

We now establish the existence of building processes which satisfy
property \ref{pr1}.

\textbf{Brushing} $M_{B}$ (Fig.5). Let $T$ be a rooted tree with a
main stem $\sigma $. We are going to associate a rooted brush $B$
with it, denoted $M_B \left( {T,\sigma } \right)$ of the same size
built from the same main stem $\sigma $ with the following
process: every $e$-branch $B_e $ connected to the main stem by
edge $e$ is replaced by a path of length $\left| {E\left( {B_e }
\right)} \right|$ connected by the same edge.

\begin{figure}[htbp]
\centerline{\includegraphics[width=2.92in,height=1.48in]{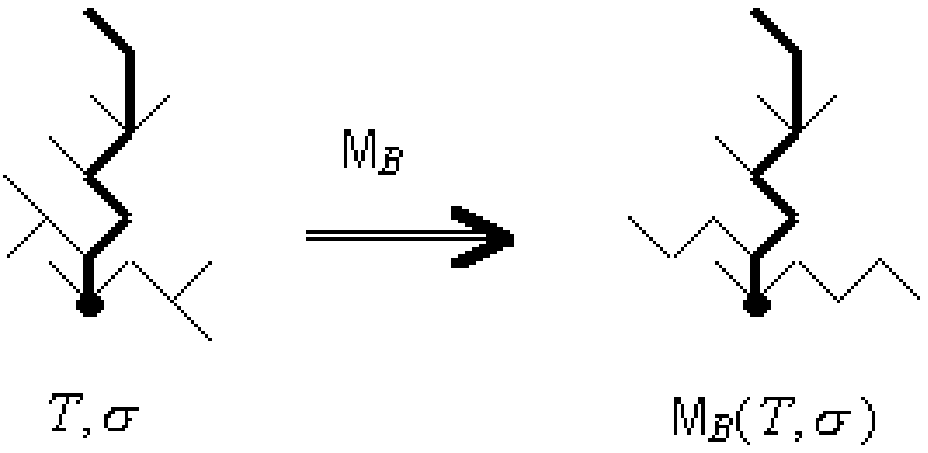}}
\label{fig5}
\begin{center}
Fig.5.
\end{center}
\end{figure}

\textbf{Ramifying }$M_T^k $\textbf{.} For the second building
process we work in two steps :

\textbf{First step.} Given rooted trees $T_1 ,...,T_k $ with
disjoint vertex sets, we build another rooted tree $T$, denoted
$\left[ {T_1 ,...,T_k } \right]$, in the following way~:

\[
V(T) = \bigcup\limits_{i = 1}^k {V\left( {T_i } \right)} \cup \left\{ {v_1
,...,v_{k + 1} } \right\},
\]

\[
E(T) = \bigcup\limits_{i = 1}^k {E\left( {T_i } \right)} \cup \left\{
{\left\{ {v_1 ,r\left( {T_1 } \right)} \right\},...,\left\{ {v_k ,r\left(
{T_k } \right)} \right\}} \right\} \cup \left\{ {\left\{ {v_1 ,v_2 }
\right\},...,\left\{ {v_k ,v_{k + 1} } \right\}} \right\},
\]

\noindent
and $r(T) = v_1 $.

If $T_i = \emptyset $, conventionally $\left\{ {v_i ,r\left( {T_i
} \right)} \right\} = \emptyset $.

Prosaically, from a path $P_k = \left[ {v_1 ,...,v_{k + 1} }
\right]$ of size $k$ and from $k$ rooted trees $T_1 ,...,T_k $, we
build a rooted tree joining a branch $T_i $ to the vertex $v_i $
of $P$ (Fig.6).

\begin{figure}[htbp]
\centerline{\includegraphics[width=1.73in,height=1.19in]{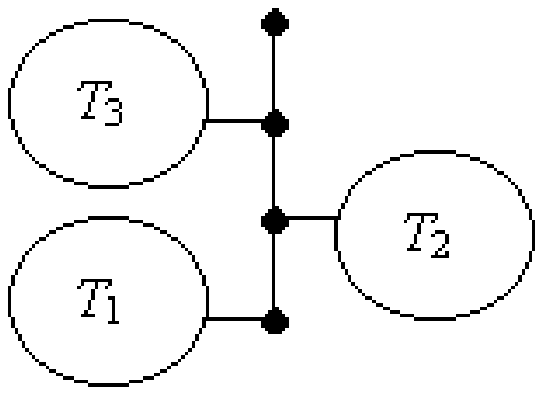}}
\label{fig6}
\begin{center}
Fig.6. A rooted comb $\left[ {T_1 ,T_2 ,T_3 } \right]$
\end{center}
\end{figure}

\textbf{Second step.} By convention, $P_{ - 1} = \emptyset $.

We are going to construct rooted trees $T_k $ in the following
way~:\\ $T_{ - 1} = \emptyset$, $T_0 = X_0$, and $\forall i,$ $1
\le i \le k, T_i = \left[ {T_{\min \left( {u_1 ,i - 1} \right)}
,...,T_{\min \left( {u_{n_i } ,i - 1} \right)} } \right]$ if $X_i
= \left[ {P_{u_1 } ,...,P_{u_{n_i } } } \right]$.

We can now define $M_T^k$~:

$$M_T^k \left( {\left( {X_n } \right)_{n \in \mathbb{N}} } \right)
= T_k.$$

\begin{lemma} The building processes described above verify the
property \ref{pr1}.
\end{lemma}

\begin{proof} First, note that $M_T \left( {\left( {X_n } \right)_{n \in \mathbb{N}} } \right)$
is an increasing sequence. We prove the lemma by recurrence on the
size $m$ of $T$. When $m = 0$ or $m = 1$, this is trivial. We
suppose the property is verified for $T$ with size $m < m_0 $. Let
$T$ be a rooted tree of size $m_0 $ with a stem $\sigma $, we note
$e_1 ,...,e_k $ the edges of $T$ issued from $\sigma $ which do
not belong to $\sigma $. To each $e$-branch of $T$ with $e \in
\left\{ {e_1 ,...,e_k } \right\}$ corresponds by $M_B $ a
$e$-branch (it is a path of same size) in $M_B \left( {T,\sigma }
\right)$. So there exists $k$ distinct $e$-branches $R_1 ,...,R_k
$ in $X_{m_0 } $ that we can respectively contract to obtain each
$e$-branch with $e = e_1 ,...,e_k $ in $M_B \left( {T,\sigma }
\right)$. By recurrence hypothesis, we have for $1 \le i \le k,
B_{e_i } \preceq M_T^{\left| {E\left( {B_{e_i } } \right)}
\right|} \left( {\left( {X_n } \right)_{n \in \mathbb{N}} }
\right)$ and we have also $M_T^{\left| {E\left( {B_{e_i } }
\right)} \right|} \left( {\left( {X_n } \right)_{n \in \mathbb{N}}
} \right)\preceq M_T^{\left| {E\left( {R_i } \right)} \right|}
\left( {\left( {X_n } \right)_{n \in \mathbb{N}} } \right)$. So
each $e$-branch of $T$ is a minor contraction of $M_T^{\left|
{E\left( {R_i } \right)} \right|} \left( {\left( {X_n } \right)_{n
\in \mathbb{N}} } \right)$. By associativity of contraction map,
we have $T\preceq M_T^{\left| {E\left( T \right)} \right|} \left(
{\left( {X_n } \right)_{n \in \mathbb{N}} } \right)$.\qed
\end{proof}

In this phase, we determine a sequence of rooted combs $\left(
{X_i } \right)_{i \in \mathbb{N}} $ such that the rooted combs
$X_i $ are strongly $i$-universal for the rooted brushes.

In order to achieve this result, we define $F_p $ as the set of
functions $f$~: $\left\{ {1,...,p} \right\} \to \left\{
{1,...,\left\lfloor {\frac{p}{2}} \right\rfloor } \right\}$
satisfying the following property~:

\[
\left( {\forall n \in \left\{ {1,...,p} \right\}} \right)\left(
{\forall i \le \left\lfloor {\frac{n}{2}} \right\rfloor }
\right)\left( {\exists k \in \mathbb{N}} \right)\left( {n - i + 1
\le k \le n\mbox{ and }f(k) \ge i} \right)
\]

\begin{lemma} $F_p $ is not empty, it contains the following function
$\varphi _p $, defined for $1 \le i \le p$ by~:

\[
\varphi _p \left( i \right) = \min \left( {2^{\upsilon _2 \left( i \right) +
1} - 1,\left\lfloor {\frac{p}{2}} \right\rfloor ,i - 1} \right)
\]

\noindent
where $\upsilon _2 \left( k \right)$ is the 2-valuation of $k$ (i.e. the
greatest power of 2 dividing $k)$.
\end{lemma}

\begin{proof}The verification is obvious.\qed
\end{proof}

\begin{lemma}\label{lm:b} For every sequence $F = \left( {f_1 ,f_2 ,...} \right)$
of functions such that $f_i \in F_i $ for $i \ge 1$ and $f_i \left( k
\right) \le f_{i + 1} \left( k \right)$ for all $i \ge 1$ and $1 \le k \le
i$, the rooted comb defined by $Comb_m^F = \left[ {Pf_1^m ,...,Pf_m^m }
\right]$ where $Pf_i^m $ designs the path of size $f_m (m + 1 - i) - 1$, for
$1 \le i \le m$ is strongly $m$-universal for the rooted brushes.
\end{lemma}

\begin{proof} By induction on $m$ : $Comb_1^F $ is strongly
1-universal for the rooted brushes.

Suppose that $Comb_i^F $ has all rooted brushes with $i - 1$ edges
as rooted contractions.

We consider two cases depending on the shape of a rooted brush $B$ of size
$i$ :

\hspace{1.7cm} case 1 \hspace{4cm} case 2
\begin{figure}[htbp]
\centerline{\includegraphics[width=3.80in,height=2.05in]{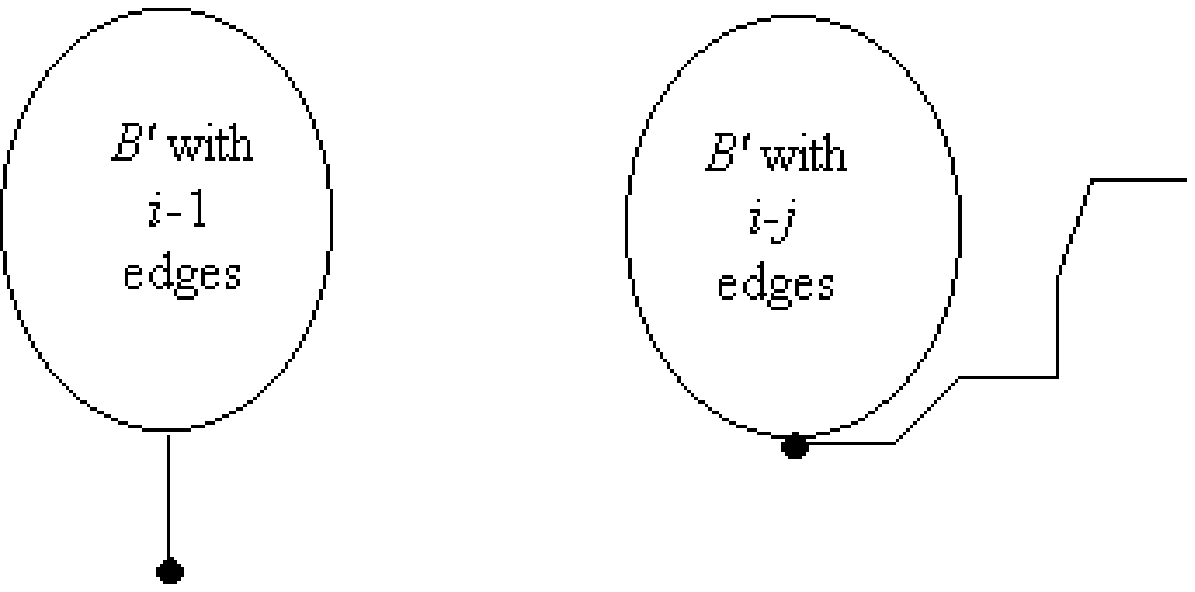}}
\end{figure}

Brushes of case 1 are clearly rooted contractions of the rooted
comb $Comb_i^F $ (${B}'\preceq Comb_{i - 1}^F $, so $B\preceq
\left[ {P_0 ,Pf_1^{i - 1} ,...,Pf_{i - 1}^{i - 1} } \right]\preceq
Comb_i^F )$. Let us study case~2 : $B'$ is by induction hypothesis
a rooted contraction of the rooted comb $Comb_{i - j}^F $,
moreover $Comb_{i - j}^F \preceq \left[ {Pf_{j + 1}^i ,...,Pf_i^i
} \right].$ Finally, by the property of $f_i $, there exists $1
\le \alpha \le j$, such that $Pf_\alpha ^i $ has more than $j$
edges. Linking these two points, we can conclude that the rooted
brush $B$ is always a rooted contraction of the rooted comb
$Comb_i^F $.\qed
\end{proof}

The rooted comb built as in lemma \ref{lm:b} will be said to be
\textit{associated to the sequence} $F$ and denoted by $Comb_m^F
$.

\begin{theorem} A minimum strongly $m$-universal rooted brush for the
rooted brushes has $O(m\ln (m))$ edges.
\end{theorem}

\begin{proof} Proceeding as for theorem \ref{th1}, we obtain, mutatis mutandis, that a
$m$-universal brush for the brushes has at least $m\ln (m) + O(m)$
edges. This order of magnitude is precisely the size of the
strongly $m$-universal rooted comb $Comb_m^F$ for the class of
rooted brushes.\qed
\end{proof}

We have this immediate corollary~:

\begin{corollary} A minimum $m$-universal brush for the brushes has
$O(m\ln (m))$ edges.
\end{corollary}

By convention, we put $Comb_0^F = P_0 $ (tree reduced in a vertex)

We define $Tree_m^F = M_T^m \left( {\left( {Comb_n^F } \right)_{n
\in \mathbb{N}} } \right)$.

As before, we will say that the tree built in such a way is\textit{ recursively associated to the sequence }$F$ and denoted by
$Tree_m^F $.

Thus, we have :

\begin{theorem} The rooted tree $Tree_m^F $ is strongly $m$-universal
for the class of rooted trees.
\end{theorem}

We now analyze the size of $Tree_m^F $.

\begin{proposition} Let $F = \left( {f_1 ,f_2 ,...} \right)$ be a
sequence of functions such that $f_i \in F_i $ for $i \ge 1$. The
size of a $m$-universal tree constructed from the sequence is
given by the following recursive formula :

$u_{ - 1} = - 1, u_0 = 0$ and $u_k = 2k - 1 + \sum\limits_{i =
1}^k {u_{f_k \left( i \right) - 1} } $
\end{proposition}

\begin{proof}It derives from the following observation~:\\
\noindent $m$ edges constitute the main stem, we have to add $m -
1$ edges to link branches to the main stem and $\sum\limits_{i =
1}^k {u_{f_k \left( i \right) - 1} } $ edges for the branches.\qed
\end{proof}

\begin{theorem} There is a sequence of functions $G = \left( {g_1
,g_2 ,...} \right)$ such that $g_i \in F_i $ and $\left| {E\left( {Tree_m^G }
\right)} \right| < \left( {2m} \right)^c$ where $c = 1.984...$ is the unique
positive solution of the equation $\frac{1}{2^c} + \frac{1}{2^{2c}} +
\frac{1}{2^{\left( {c - 1} \right)} - 1} - \frac{1}{2^c - 1} = 1$.
\end{theorem}

\begin{proof}We take the following sequence of functions~:\\
$g_m \left( i \right) = \min \left( {2^{\upsilon _2 \left( i
\right) + 1},i} \right)$ if $i < m$ and $i$ even, $g_m \left( i
\right) = 1$ if $i$ odd and $g_m \left( m \right) = \left\lfloor
{\frac{m}{4}} \right\rfloor $. It is clear that, if $m$ is a power
of 2, the comb $Comb_m^G $ is strongly $m$-universal for the
brushes.

In fact, the function $g_m $ takes the value $2^{\upsilon _2 \left( i
\right) + 1}$ when $i$ is not a power of 2, otherwise it is equal to $i$.
Thanks to this remark and with $u_m < m + \sum\limits_{i = 1}^m {u_{f_m
\left( i \right)} } $, (the sequence of sizes is increasing), we obtain
$u_{2^n} < 2^n + 2^{n - 1} + \sum\limits_{i = 2}^{n - 1} {2^{n - i}u_{2^i} }
- \sum\limits_{i = 2}^{n - 1} {u_{2^i} } + u_{2^{n - 1}} + u_{2^{n - 2}} $.
Thus, in evaluating the sums and reorganizing the terms, we obtain :

\[
u_{2^n} < \alpha _n + 2^{nc}\beta
\]

\noindent
with

\[
\alpha _n = 2^{n - 1} + 1 + 2^c + \frac{1}{2^c - 1} - \left(
{\frac{2^n}{2^{\left( {c - 1} \right)} - 1} + 2^{n\left( {c - 1} \right)}}
\right)
\]

\[
\beta = \frac{1}{2^c} + \frac{1}{2^{2c}} + \frac{1}{2^{\left( {c - 1}
\right)} - 1} - \frac{1}{2^c - 1}
\]\\
Now $\alpha _n < 0$ when $m > 1$ and $\beta \le 1$ by definition
of $c.$\\ So $u_{2^n} < 2^{nc}$, hence $u_m < \left( {2m}
\right)^c$.\qed
\end{proof}

\begin{remark} We observe that $c = \frac{\ln \left( x \right)}{\ln
\left( 2 \right)}$, where $x$ is the positive root of $X^4 - 5X^3 + 4X^2 + X
- 2 = 0$.
\end{remark}

\bigskip

Theorem \ref{th2} then follows since any rooted tree which is
strongly $m$-universal for the rooted trees is also clearly
$m$-universal for the class of trees.

\section{Conclusion and Related Questions}

When using the sequence $\Phi = \left( {\varphi _1 ,\varphi _2
,...} \right)$ of lemma \ref{lm:b}, the induction step leads to
involved expressions that do not allow us to find the asymptotic
behavior of the corresponding term $u_m $. A computer simulation
gives that such a $m$-universal tree for the trees has less than
$m^{1.88}$ edges. In any case, the constructive approach we
proposed here, seems to be hopeless to reach the asymptotic best
size of a $m$-universal tree for the trees.

\begin{conjecture} The minimal size of a $m$-universal tree for the
trees is $m^{1 + o\left( 1 \right)}$.
\end{conjecture}

As a possible way to prove such a conjecture, it would be
interesting to obtain an explicit effective coding of a tree of
size $m$ using a list of contracted edges taken in a $m$-universal
tree for the trees.

A variant of our problem consists in determining a minimum tree
which contains as a sub{\-}tree every tree of size $m.$ This is
closely related to a well known still open conjecture due to
Erd\"{o}s and S\"{o}s (see \cite{ES}).

\end{document}